\documentclass{article}
\usepackage[utf8]{inputenc}

\usepackage{authblk}
\usepackage[dvipdfmx]{graphicx}
\usepackage{times}
\usepackage{xcolor}
\usepackage{soul}
\usepackage[utf8]{inputenc}
\usepackage[small]{caption}
\usepackage{bm}
\usepackage{amsmath}
\usepackage{amssymb}
\usepackage{amsfonts}
\usepackage{latexsym}
\usepackage{amsthm}
\usepackage{comment}
\usepackage[ruled,vlined]{algorithm2e}
\usepackage{url}

\newtheorem{theorem}{Theorem}
\newtheorem{lemma}[theorem]{Lemma}

\newcommand{\mcal}[1]{\mathcal{#1}}
\newcommand{\mrm}[1]{\mathrm{#1}}
\newcommand{\mbf}[1]{\mathbf{#1}}

\newcommand{\ot}{\leftarrow}
\newcommand{\sem}[1]{[\![#1]\!]}

\newcommand{\lrangle}[1]{\langle #1 \rangle}

\newcommand{\suzuki}[1]{{\color[rgb]{0.3,0.6,0} #1}}
\newcommand{\kawahara}[1]{{\color{blue} #1}}

\newcommand{\marking}[1]{{\color{red} #1}}

\title{Enumerating All Subgraphs without Forbidden Induced Subgraphs via Multivalued Decision Diagrams}
\date{}
\author[1]{Jun Kawahara}
\author[2]{Toshiki Saitoh}
\author[3]{Hirofumi Suzuki}
\author[4]{Ryo Yoshinaka}
\affil[1]{Nara Institute of Science and Technology}
\affil[2]{Kyushu Institute of Technology}
\affil[3]{Hokkaido University}
\affil[4]{Tohoku University}

\begin{document}
\maketitle

\begin{abstract}
We propose a general method performed over multivalued decision diagrams that enumerates all subgraphs of an input graph that are characterized by input forbidden induced subgraphs.
Our method combines elaborations of classical set operations and the developing construction technique, called the frontier based search, for multivalued decision diagrams.
Using the algorithm, we enumerated all the chordal graphs of size at most 10 on multivalued decision diagrams.
\end{abstract}

\section{Introduction}
Enumeration is a fundamental topic in computer science.
Especially, {\it subgraph enumeration problem} is a well-studied topic.
Given a graph and constraints, the problem is to output all the subgraphs satisfying the constraints in the graph.
Several well-known techniques for enumeration have been proposed \cite{tarjan:1975, avis:1996} and applied to several graph classes.
For example, \cite{kiyomi:2006} enumerates all subgraphs belonging to the class of {\it chordal graphs} based on reverse search~\cite{avis:1996}.
These traditional algorithms enumerate subgraphs one by one explicitly and take time depending on the number of output subgraphs.
Unfortunately, the number of output subgraphs is exponentially huge in the size of the input graph.

On the other hand, to approach the subgraph enumeration problem by implicit enumeration, techniques constructing a compressed representation such as {\it Zero-suppressed Binary Decision Diagram} (ZDD) \cite{minato:1993} are well-studied.
Computation time of such techniques does not depend on the number of subgraphs but the size of the constructed ZDDs.
A classical technique such as the one proposed in~\cite{coudert:1997} is to use {\it apply operation} and {\it family algebra} that are useful function of ZDD.
Moreover, a novel algorithm named {\it frontier-based search} (FBS)  \cite{kawahara:2017} has been developed recently.
FBS has been applied for enumerating various classes of subgraphs such as paths, cycles, forests, partitions, and so on.
However, graph classes handled by ZDD based techniques are limited to rather simple ones only.

In this paper, we propose a general technique that enumerates subgraphs belonging to graph classes characterized by {\it forbidden induced subgraphs}.
Several graph classes such as chordal, interval, split, and threshold graphs are characterized by rather simple forbidden induced subgraphs, like cycles, paths, and their complements.
For example, a graph is called chordal if and only if it has no cycles of size at least 4 as a vertex induced subgraph.
The proposed method needs to be given a ZDD representing forbidden induced subgraphs, which we assume to be computed by an existing method or some way.
Our technique consists of the following three steps, which involve FBS and family algebra over \emph{multivalued decision diagrams (MDD)}~\cite{KAM:1998}:
\begin{enumerate}
\item Enumerating forbidden induced subgraphs on a ZDD (in some way);
\item Adding edges induced by the forbidden induced subgraphs as an MDD by FBS;
\item Constructing a ZDD enumerating subgraphs that avoid forbidden subgraphs as induced ones by a novel operation of family algebra.
\end{enumerate}
The frist step depends on the target graph class but the last two steps do not.
This paper describes those two steps.

As a demonstration of our method, we enumerated chordal graphs by experiments.
We succeeded in enumerating all the 215,488,096,587 chordal graphs of size 10 as a ZDD.

\section{Preliminaries}

\subsection{Forbidden induced subgraphs}

Let $G = (V,E)$ be a graph with a vertex set $V$ and an edge set $E \subseteq \{\{u,v\} \mid u,v \in V\}$.
For any vertex subset $U \subseteq V$, $E[U]$ denotes the set of edges whose end points are both included in $U$, i.e., $E[U] = \{\,e \in E \mid e \subseteq U\,\}$, called {\em induced edges} ({\it by $U$}).
For any edge subset $D \subseteq E$, $\bigcup D$ denotes the set of the end points of each edge in $D$, i.e., $\bigcup D = \bigcup_{\{u,v\} \in D} \{u,v\}$, called {\em induced vertices} ({\em by $D$}).
We call $(U, E[U])$ the ({\it vertex}) {\em induced subgraph} ({\em by $U$}).
Let $G[D] = (\bigcup D, D)$, called the {\em edge induced subgraph} ({\em by $D$}).
This paper often identifies an edge induced subgraph $G' = (\bigcup D, D)$ and the edge set $D$.

Some graph classes are characterized by forbidden subgraphs.
We say that a graph class $\mcal{G}$ is \emph{FIS-characterized} by a graph class $\mcal{F}$
 if $\mcal{G}$ consists of graphs $G=(V,E)$ such that none of the vertex subsets of $V$ induces a graph belonging to $\mcal{F}$, i.e.,
\[
	(V,E) \in \mcal{G} \iff \forall U \subseteq V,\, (U,E[U]) \notin \mcal{F}
\,.\]
For example, the class of chordal graphs is FIS-characterized by the class of cycles of size at least 4.

\subsection{Multi-valued Decision Diagrams}
A \emph{$k$-colored subset} of a finite set $E$ is a $k$-tuple $\vec{D}=(D_1,\dots,D_k)$ of subsets $D_i \subseteq E$ such that $D_i \cap D_j = \emptyset$ for any distinct $i$ and $j$.
To represent and manipulate sets of $k$-colored subsets, we use \emph{$k$-valued decision diagrams ($k$-DDs)},
which are special types of \emph{multi-valued decision diagram} with lack of a reduction rule.

A $k$-DD over a finite set $E = \{e_1,\dots,e_m\}$ is a labeled rooted directed acyclic graph $\mbf{Z}=(N,A,\ell)$ with a node set $N$, an arc set $A$ and a labeling function $\ell$.
The node set $N$ has exactly one root node $\rho$ and exactly two terminal nodes $\bot$ and $\top$.
Each non-terminal node $\alpha \in N\setminus\{\top,\bot\}$ has a label $\ell(\alpha) \in \{1,\dots,m\}$ and has exactly $k+1$ outgoing arcs called $0$-arc, $1$-arc, \dots, and $k$-arc.
The node pointed at by the $j$-arc of $\alpha$ is called the $j$-child and denoted by $\alpha_j$ for each $j \in \{0,1,\ldots,k\}$.
It is satisfied that $\ell(\alpha_j) = \ell(\alpha) + 1$ if $\alpha_j$ is not a terminal.

Each path $\pi$ in a $k$-DD represents a $k$-colored subset $\sem{\pi} = (D_1,\dots,D_k)$ of $E$ defined by
\[
	D_j = \{\, e_{\ell(\beta)} \mid \text{$\pi$ includes the $j$-arc of $\beta$} \,\}\,,
\]
for $j \in \{1,\dots,k\}$.
The $k$-DD $\mbf{Z}$ itself represents a set of $k$-colored subsets
\[
\sem{\mbf{Z}} = \{\, \sem{\pi} \mid \pi \text{ is a path from the root $\rho$ to the terminal $\top$}\,\}
\,.\]

We call a $k$-DD \emph{reduced} if there are no distinct nodes $\alpha$ and $\beta$ such that $\ell(\alpha)=\ell(\beta)$ and $\alpha_j=\beta_j$ for all $j \in \{0,\dots,k\}$.
If a $k$-DD has nodes that violate this condition, those can be merged repeatedly until it becomes reduced.
This reduction does not change the semantics of the $k$-DD.

We remark that $k$-DDs, 2-DDs, and 3-DDs are almost identical to MDDs, binary decision diagrams as well as zero-suppressed binary decision diagrams, and ternary decision diagrams, respectively,
except a reduction rule that eliminates nodes so that the obtained data structure will be more compact. 
It is possible for our algorithm with slight modification to handle ``zero-suppress'' $k$-DDs, where a node can be eliminated if all the $j$-children for $1 \le j \le k$ point at the terminal $\bot$.
However, for simplicity, we have defined $k$-DDs without employing such a reduction rule, where the label of a child node is always bigger than the parent's by one.

\section{Proposed Algorithm}\label{sec:algorithm}
Suppose that a graph class $\mcal{G}$ is FIS-characterized by $\mcal{F}$.
In this section, we present an algorithm that enumerates all the subgraphs of an input graph $G$ belonging to $\mcal{G}$ provided that all the forbidden induced subgraphs of $G$ belonging to $\mcal{F}$ is also given as an input.
Hereafter we fix an input graph $G=(V,E)$ and restrict $\mcal{G}$ and $\mcal{F}$ to be the subgraphs of $G$.
Here by a subgraph of $G$ we mean a graph $G'=(V,D)$ for some $D \subseteq E$.
Therefore, we may identify a graph and its edge set.
That is, $\mcal{G}$ and $\mcal{F}$ are represented as sets of subsets of $E$.
The set $\mcal{F}$ is given as a 2-DD.

Recall the definition of FIS-characterization:
\[
	D \in \mcal{G} \iff \forall U \subseteq \bigcup D,\, D[U] \notin \mcal{F}
\,.\]
In other words, if a graph $D \in \mcal{G}$ contains a subgraph $D'$ belonging to $\mcal{F}$, then $D'$ must induce edges with which it does not belong to $\mcal{F}$.
Our proposed method can be divided into two phases.
In the first phase, as its details will be described in Section~\ref{sec:edge-induction}, we construct a 3-DD $\mbf{I}$ for the set of $2$-colored subsets $(F_1,F_2)$ of $E$ such that $F_1 \in \mcal{F}$ and $F_2 = E[\bigcup F_1]  \setminus F_1$.
Then we construct a 2-DD $\mbf{Z}$ for the set of edge sets $D$ such that 
\[
\forall (F_1,F_2) \in \sem{\mbf{I}},\, ( F_1 \subseteq D \implies F_2 \cap D \neq \emptyset )
\,,\]
as described in Section~\ref{sec:goal}.
The following lemma ensures that our method indeed gives the desired subgraphs.
\begin{lemma}\label{lem:correctness}
Let $D \subseteq E$. The following conditions are equivalent:
\begin{enumerate}
\item $\forall U \subseteq \bigcup D,\, D[U] \notin \mcal{F}$,
\item $\forall F \in \mcal{F},\, ( F \subseteq D \implies (E[\bigcup F]  \setminus F) \cap D \neq \emptyset )$.
\end{enumerate}
\end{lemma}
\begin{proof}
($\Leftarrow$)
Suppose that $D$ does not satisfy the first condition.
There is $U \subseteq \bigcup D$ such that $(U,D[U]) \in \mcal{F}$.
Let $F = D[U]$. 
By definition $F \subseteq D$ and $E[\bigcup F] \setminus F \subseteq E[U] \setminus D[U] \subseteq E \setminus D$.
We have $E[\bigcup F] \setminus F \cap D = \emptyset$.

($\Rightarrow$)
Suppose that  $D$ does not satisfy the second condition.
There is $F \in \mcal{F}$ such that $F \subseteq D$ and $(E[\bigcup F]  \setminus F) \cap D = \emptyset$.
We will show that $U = \bigcup F$ disproves the first statement.
Since $F \subseteq D$, clearly $U \subseteq \bigcup D$.
It suffices to show that $F = D[U]$, which implies $D[U] \in \mcal{F}$.
Clearly 
\[
F = F[\bigcup F] \subseteq D[\bigcup F] =  D[U]
\,.\]
By assumption,
\[
D[U] \setminus F = (E [U] \cap D) \setminus F = (E[U] \setminus F) \cap D = \emptyset
\]
and thus $D[U] \subseteq F$.
\end{proof}

\subsection{Edge induction}\label{sec:edge-induction}
This subsection presents an algorithm that 
gives a 3-DD $\mbf{I}$ representing $\sem{\mbf{I}} =\{\,(F,\allowbreak{}E[\bigcup F]  \setminus F) \mid F \in \sem{\mbf{F}} \,\}$
from an input 2-DD $\mbf{F}$ representing a set $\sem{\mbf{F}}$ of edge sets.
That is, we ``color'' the edges induced by $\bigcup F$ with the second color for each $F \in \mcal{F}$.
Our algorithm can be seen as an instance of the so-called \emph{frontier-based search}, which is a generic framework for enumerating all the subgraphs with a specific property from an input graph.
Algorithm~\ref{alg:edge-induction} constructs a 3-DD in a top-down manner, where the initial 3-DD has only the root node $\rho_\mbf{I}$ with $\ell(\rho_\mbf{I})=1$.
By giving children to already constructed nodes, we expand the diagram.
Each node $\alpha$ of the diagram under construction has auxiliary information called \emph{configuration},
which is a pair $(n_\alpha,f_\alpha)$ of a node $n_\alpha$ of the input 2-DD $\mbf{F}$ and a map $f_\alpha$ from a subset $E^{\lrangle{ \ell(\alpha)}} = \bigcup E^{\ge  \ell(\alpha)} \cap \bigcup E^{<  \ell(\alpha)}$ of $V$ to $\{-1,0,1,2\}$, where $E^{\ge i} = \{e_i,\dots,e_m\}$ and $E^{< i} = \{e_1,\dots,e_{i-1}\}$.
No distinct nonterminal nodes have the same configuration.
The first component $n_\alpha$ satisfies the property that for any path $\pi$ from the root $\rho_\mbf{I}$ to $\alpha$ in $\mbf{I}$, there is a path $\theta$ from the root $\rho_\mbf{F}$ to $n_\alpha$ in $\mbf{F}$ such that $\sem{\pi} = (\sem{\theta},F_2)$ for some $F_2 \subseteq E$ (but not vice versa).
The default value of $f_\alpha$ is set to $f_\alpha(u)=0$ for all $u \in E^{\lrangle{i}}$.
If it has non-zero value, 
\begin{itemize}
	\item $f_\alpha(u)=-1$ means that $u$ must not occur in $\sem{\pi}$,
	\item $f_\alpha(u)=1$ means that there is no $v'$ such that $\{u,v'\}$ is colored 1 in $\sem{\pi'}$ but
	there must be $v$ such that $\{u,v\}$ is colored 1 in $\sem{\pi}$,
	\item $f_\alpha(u)=2$ means that there is $v'$ such that $\{u,v'\}\in E^{< i}$ is colored 1 in $\sem{\pi'}$,
\end{itemize}
 for  any path $\pi$ from the root $\rho_\mbf{I}$ to the terminal $\top$ passing through $\alpha$ and any path $\pi'$ from $\rho_\mbf{I}$ to $\alpha$.
The algorithm starts with the root node $\rho_\mbf{I}$ with configuration $(\rho_\mbf{F},\emptyset)$.
\begin{algorithm}[ht]
\caption{Inducing edges}
\label{alg:edge-induction}
	\SetKwInOut{Input}{input}\SetKwInOut{Output}{output}
	\SetKwFunction{Child}{Child}
	 \Input{a 2-DD $\mbf{F}$ (representing forbidden induced subgraphs)}
	 \Output{a 3-DD $\mbf{I}$ (coloring the edges induced by the forbidden graphs)}
	 \BlankLine
	let $N_1 \leftarrow \{(\rho_\mbf{F},\emptyset)\}$, $N_i \leftarrow \emptyset$ for $i=2,\ldots,m$ and $N_{m+1} \ot \{\top,\bot\}$\;
	\For{$i=1,\ldots,m$}{
		\For{each $\alpha \in N_i$}{
			\For{$j = 0,1,2$}{
				let $\alpha_j \ot \Child(\alpha,j)$\;
				\lIf{$\alpha_j \notin N_{i+1}$}{add a new node $\alpha_j$ with label $i+1$ to $N_{i+1}$}
				let $\alpha_j$ be the $j$-child of $\alpha$\;
			}
		}
	}
	\textbf{return} the 3-DD consisting of nodes of $N_1,\dots,N_{m+1}$\;
\end{algorithm}
\begin{algorithm}[ht]
\caption{$\texttt{Child}(\alpha,j)$}
\label{alg:edge-induction-child}
	\LinesNumbered
	\SetKwInOut{Input}{input}\SetKwInOut{Output}{output}
	\SetKwFunction{Child}{Child}
	 \Input{node $\alpha$ with configuration $(\beta,f)$ and a child number $j$}
	 \Output{configuration of the $j$-th child $\alpha_j$ of $\alpha$}
	 \BlankLine
	let $i \ot \ell(\alpha)$ and $\{u_1,u_2\} \ot e_i$\;
	\lIf{$j=1$}{let $n_{j} \ot \beta_1$}
	\lElse{let $n_{j} \ot \beta_0$}
	\lFor{all $v \in E^{\lrangle{i}}$}{let $f_{j}(v) \ot f(v)$}
	\lFor{all $v \in E^{\lrangle{i+1}} \setminus E^{\lrangle{i}}$}{let $f_{j}(v) \ot 0$}
	\If{$j=0$}{\label{algline:child00}
		\lIf{$n_0 = \bot$ or $f(u_1) f(u_2) \ge 1$}{let $\alpha_0 \ot \bot$}
		\lElseIf{$f(u_1) \ge 1$}{let $f_{0}(u_2) \ot -1$}
		\lElseIf{$f(u_2) \ge 1$}{let $f_{0}(u_1) \ot -1$}
	\label{algline:child01}}
	\Else{\label{algline:child10}
		\lIf{$n_j = \bot$ or $f(u_1) = -1$ or $f(u_2) = -1$}{let $\alpha_j \ot \bot$}
		\Else{
			\For{$k=1,2$}{
				\If{$f(u_k) = 0$}{
					\For{all $v \in E^{\lrangle{i}}$ such that $\{u_k,v\} \in E^{< i}$}{
						\lIf{$f(v) \ge 1$}{let $\alpha_j \ot \bot$}
						\lElse{let $f_j(v) \ot -1$}
					}
				}
			{let $f_j(u_k) \ot \max\{f(u_k),3-j\}$\;}\label{algline:child11}
			}
		}
	}
	\lIf{$u_k \notin \bigcup E^{\ge i+1}$ and $f_j(u_k) = 1$ for some $k \in \{1,2\}$}{let $\alpha_j \ot \bot$}\label{algline:child3}
	\lIf{$\alpha_j \neq \bot$}{let $\alpha_j \ot (n_j,f_j \upharpoonright E^{\lrangle{i+1}})$}\label{algline:child4}
	\textbf{return} $\alpha_j$\;
\end{algorithm}

Algorithm~\ref{alg:edge-induction-child} gives the configuration of the $j$-child of a node $\alpha \in N_i$, unless the child must be a terminal.
Let the configuration of $\alpha$ be $(\beta,f)$ and $e_i = \{u_1,u_2\}$.

Choosing the 0-arc of $\alpha$ means that we do not include the edge $e_i$ in a 2-colored graph under consideration.
Recall that if both $u_1$ and $u_2$ are used in a graph, then the edge $e_i$ must be colored 1 or 2.
Lines~\ref{algline:child00}--\ref{algline:child01} reflect this restriction.

Choosing the $j$-arc with $j \ge 1$ means that the edge $e_i$ is colored $j$ in the resultant 2-colored graph.
This case is handled on Lines~\ref{algline:child10}--\ref{algline:child11}.
This clearly contradicts $f(u_k)=-1$ for any of $k \in \{1,2\}$, which means that $u_k$ must not be used.
If $f(u_k)=0$, this means that so far all edges $\{u_k,v\} \in E^{<i}$ are colored 0, i.e., they do not occur in 2-colored subgraphs under consideration.
On the other hand, if $f(v) \ge 1$, this means that $v$ will occur together with $u$.
This contradicts that the edge $\{u_k,v\}$ remains colored 0.
If $e_i =\{u_1,u_2\}$ is colored $1$, this means that it is in a forbidden graph, so we let $f_j(u_k)=2$.
If $e_i =\{u_1,u_2\}$ is colored $2$, this means that it is induced by an vertex in a forbidden graph, so we let $f_j(u_k)=1$ unless $f_j(u_k) = 2$.

In addition, if $f(u_k)=1$ and $u_k \notin \bigcup E^{\ge i+1}$, this means that $u_k$ is supposed to have an edge colored with $1$ but we have decided not to  color any edges connecting $u_k$ with 1.  This restriction is checked on Line~\ref{algline:child3}.

Since we do not need to remember the values of $f_j(v)$ for $v \notin E^{\lrangle{i+1}}$ in the further computation, we restrict the domain of $f_j$ to be $E^{\lrangle{i+1}}$ on Line~\ref{algline:child4}.

\subsection{Enumeration of subgraphs with no forbidden induced subgraphs}\label{sec:goal}
We now give an operation that computes a 2-DD $\mbf{D}$ for 
\[
	\sem{\mbf{D}}=\chi(\sem{\mbf{I}})=\{\, D \subseteq E \mid  \forall (F_1,F_2) \in \sem{\mbf{I}},\, ( F_1 \subseteq D \implies F_2 \cap D \neq \emptyset ) \,\}
\]
from a 3-DD $\mbf{I}$.
Note that the domain of $\chi$ is 2-colored subsets of $E$ and the codomain is (1-colored) subsets of $E$.
When $\mbf{I}$ represents $\sem{\mbf{I}} = \{\,(F,E[\bigcup F]  \setminus F) \mid F \in \sem{\mbf{F}} \,\}$ for a set $\sem{\mbf{F}}$ of forbidden graphs,
 we obtain the FIS-characterized set, by Lemma~\ref{lem:correctness}.
We compute $\mbf{D}$ from a 3-DD ${\mbf{I}}$ in a bottom-up recursive manner.

Here we give a semantics of a node $\alpha$ of a $k$-DD by
\[
	\sem{\alpha} = \{\, \sem{\pi} \mid \text{$\pi$ is a path from $\alpha$ to $\top$}\,\}
\,.\]
Clearly $\sem{\mbf{D}} = \sem{\rho_\mbf{D}}$ for any $k$-DD $\mbf{D}$ and its root $\rho_\mbf{D}$.
For a set $\mcal{I}$ of 2-colored subsets of $E^{\ge i}$, define 
\[
	\chi_i(\mcal{I})=\{\, D \subseteq E^{\ge i} \mid  \forall (F_1,F_2) \in \mcal{I},\, ( F_1 \subseteq D \implies F_2 \cap D \neq \emptyset ) \,\}
\,.\]
According to this definition, the base of the recursion is given by 
\begin{itemize}
\item $\chi_{m+1}(\sem{\bot}) = \chi_{m+1}(\emptyset) = \{\emptyset\}$,
\item $\chi_{m+1}(\sem{\top}) = \chi_{m+1}(\{\emptyset,\emptyset\}) = \emptyset$.
\end{itemize}
For $i \le m$, it holds that
\begin{align}
\chi_i(\mcal{I}) &= \left( \chi_{i+1}(\mcal{I}_0) \cap \chi_{i+1}(\mcal{I}_2) \right)
\cup \left( e_i * \left( \chi_{i+1}(\mcal{I}_0) \cap \chi_{i+1}(\mcal{I}_1) \right) \right)\,,
\label{eq:recursion}
\end{align}
where
\begin{align*}
e * \mcal{D} &= \{\, \{e\}\cup D \mid D \in \mcal{D}\,\} \text{ for any family $\mcal{D}$ of (1-colored) subsets of $E$,}
\\
\mcal{I}_0 &= \{\, (F_1,F_2) \in \mcal{I} \mid e_i \notin F_1 \cup F_2\,\},
\\
\mcal{I}_1 &= \{\, (F_1 \setminus \{e_i\},F_2) \mid e_i \in F_1,\ (F_1,F_2) \in \mcal{I} \,\},
\\
\mcal{I}_2 &= \{\, (F_1,F_2 \setminus \{e_i\}) \mid e_i \in F_2,\ (F_1,F_2) \in \mcal{I} \,\}.
\end{align*}
That is, if $\mcal{I}=\sem{\alpha}$, then  $\mcal{I}_0=\sem{\alpha_0}$, $\mcal{I}_1=\sem{\alpha_1}$ and $\mcal{I}_2=\sem{\alpha_2}$.
If a 2-DD has a node $\beta$ with label $e_i$ such that $\sem{\beta} = \chi_i(\mcal{I})$, then $\sem{\beta_0} =  \chi_{i+1}(\mcal{I}_0) \cap \chi_{i+1}(\mcal{I}_2)$ and  $\sem{\beta_1} = \chi_{i+1}(\mcal{I}_0) \cap \chi_{i+1}(\mcal{I}_1) $.

Equation~(\ref{eq:recursion}) is justified by the following observation.
Let us partition $\chi_i(\mcal{I})$ into two depending on whether a set contains $e_i$, i.e., $\chi_i(\mcal{I}) = \mcal{D}_0 \cup (e_i * \mcal{D}_1)$ where no sets in $\mcal{D}_0 \cup \mcal{D}_1$ contain $e_i$.
Then by definition,
\begin{align*}
\mcal{D}_0 &= \{\, D \subseteq E^{\ge i} \mid e_i \notin D \wedge \forall (F_1,F_2) \in \mcal{I}, \left( F_1 \subseteq D \implies F_2 \cap D \neq \emptyset \right)\,\} 
\\
 &= \{\, D \subseteq E^{\ge i+1} \mid \forall (F_1,F_2) \in \mcal{I}_0,\, \left( F_1 \subseteq D \implies F_2 \cap D \neq \emptyset \right)
\\ & \text{\phantom{$= \{\, D \subseteq E^{\ge i+1} \mid$}}
\wedge  \forall (F_1,F_2) \in \mcal{I}_1 ,\, \left(\{e_i\} \cup F_1 \subseteq D \implies F_2 \cap D \neq \emptyset \right)
\\ & \text{\phantom{$= \{\, D \subseteq E^{\ge i+1} \mid$}}
\wedge  \forall (F_1,F_2) \in \mcal{I}_2 ,\, \left( F_1 \subseteq D \implies (\{e_i\} \cup F_2) \cap D \neq \emptyset \right)
 \,\}.
\end{align*}
For $D \subseteq E^{\ge i+1}$, the condition $\{e_i\} \cup F_1 \subseteq D $ can never be true. In addition, $(\{e_i\} \cup F_2) \cap D \neq \emptyset$ if and only if $F_2 \cap D \neq \emptyset$. Hence,
\begin{align*}
\mcal{D}_0 &= \{\, D \subseteq E^{\ge i+1} \mid \forall (F_1,F_2) \in \mcal{I}_0 \cup \mcal{I}_2 ,\, \left( F_1 \subseteq D \implies F_2 \cap D \neq \emptyset \right)
 \,\}
 \\ &= \chi_{i+1}(\mcal{I}_0) \cap \chi_{i+1}(\mcal{I}_2)\,.
\end{align*}
On the other hand,
\begin{align*}
\mcal{D}_1 &= \{\, D \subseteq E^{\ge i+1} \mid \forall (F_1,F_2) \in \mcal{I} ,\, \left(F_1 \subseteq \{e_i\} \cup D \implies F_2 \cap (\{e_i\} \cup D) \neq \emptyset \right)\,\} 
\\
 &= \{\, D \subseteq E^{\ge i+1} \mid \forall (F_1,F_2) \in \mcal{I}_0 ,\, \left( F_1 \subseteq \{e_i\} \cup D \implies F_2 \cap (\{e_i\} \cup D) \neq \emptyset \right)
\\ & \quad\quad
\wedge  \forall (F_1,F_2) \in \mcal{I}_1 ,\,\left(\{e_i\} \cup F_1 \subseteq \{e_i\} \cup D \implies F_2 \cap (\{e_i\} \cup D) \neq \emptyset \right)
\\ & \quad\quad
\wedge  \forall (F_1,F_2) \in \mcal{I}_2 ,\,\left( F_1 \subseteq \{e_i\} \cup D \implies (\{e_i\} \cup F_2) \cap (\{e_i\} \cup D) \neq \emptyset \right)
 \,\}.
\end{align*}
Obviously, the condition $(\{e_i\} \cup F_2) \cap (\{e_i\} \cup D) \neq \emptyset$ is always true.
Recall that $e_i \notin F_1 \cup F_2$ for $(F_1,F_2) \in \mcal{I}_0$
and that  $e_i \notin F_2$ for $(F_1,F_2) \in \mcal{I}_1$.
By simplifying the formula, we obtain
\begin{align*}
\mcal{D}_1 
 &= \{\, D \subseteq E^{\ge i+1} \mid \forall (F_1,F_2) \in \mcal{I}_0 ,\,\left( F_1 \subseteq D \implies F_2 \cap D \neq \emptyset \right)
\\ & \text{\phantom{$= \{\, D \subseteq E^{\ge i+1} \mid$}}
\wedge \forall (F_1,F_2) \in \mcal{I}_1 ,\,\left(F_1 \subseteq D \implies F_2 \cap D \neq \emptyset \right)
 \,\}
 \\ &= \chi_{i+1}(\mcal{I}_0) \cap \chi_{i+1}(\mcal{I}_1)
\,.\end{align*}

Algorithm~\ref{alg:goal} computes a 2-DD for $\chi(\sem{\mbf{I}})$ from (the root of) a 3-DD $\mbf{I}$ based on Equation~(\ref{eq:recursion}).
\begin{algorithm}[t]
\caption{\text{Computing a 2-DD for $\chi(\sem{\mbf{I}})$ from a 3-DD $\mbf{I}$}}
\label{alg:goal}
\SetKwInOut{Input}{input}\SetKwInOut{Output}{output}
	 \Input{node $\alpha$ of a 3-DD}
	 \Output{node $\beta$ of a 2-DD such that $\sem{\beta}=\chi(\sem{\alpha})$}
	\lIf{$\alpha = \bot$}{\textbf{return} $\top$}
	\lElseIf{$\alpha = \top$}{\textbf{return} $\bot$}
	\lElse{
	\textbf{return} a node with label $e_{\ell(\alpha)}$ whose 0-child represents ${\chi}(\sem{\alpha_0}) \cap {\chi}(\sem{\alpha_2})$ and 1-child represents ${\chi}(\sem{\alpha_0}) \cap {\chi}(\sem{\alpha_1})$}
\end{algorithm}

\section{Experiments}

In this section, we show experimental results of constructing 2-DDs and 3-DDs for chordal graphs
to confirm the performance of our algorithm. For a given graph $G$, the 2-DD for
all the cycles on $G$ can be constructed by conventional frontier-based search~\cite{Knuth11}.
The 2-DD for all the subgraphs of $G$ that have a specified number of edges can
be constructed by the method by Kawahara et al.~\cite{kawahara:2017}.
Since both methods can be easily combined~\cite{kawahara:2017}, we can obtain the 2-DD $\mbf{F}_{\mrm{cho}}$ representing
all the cycles with size at least four on $G$. By applying the algorithm
in Sec.~\ref{sec:edge-induction} with $\mbf{F} = \mbf{F}_{\mrm{cho}}$,
we obtain the 3-DD, say $\mbf{I}_{\mrm{cho}}$, and by applying the
algorithm in Sec.~\ref{sec:goal} with $\mbf{I} = \mbf{I}_{\mrm{cho}}$,
we have the 2-DD, say $\mbf{Z}_{\mrm{cho}}$, for $\chi(\sem{\mbf{I}_{\mrm{cho}}})$, which represents the set of all the chordal subgraphs of $G$.

To see the scalability and bottleneck of our algorithm, we run it for complete (vertex-labeled) graphs with $n$ vertices. Giving a complete graph as the input means that we obtain the set of all the chordal labeled (not necessarily connected) graphs with at most $n$ vertices as a 2-DD. We implemented our algorithm in the C++ language using the TdZdd library~\cite{iwashita2013efficient} for the construction of DDs in a top-down manner. Our implementation was complied by \texttt{g++} with the \texttt{-O3} optimization option and run on a machine with Intel Xeon E5-2630 (2.30GHz) CPU and 128GB memory (Linux Centos 7.4).

Table~\ref{tab:exp_time} shows the running time and memory usage of algorithms. ``Const.\ $X$ time'' in the table indicates the time (in seconds) for constructing the (2- or 3-) DD $X$.
``Mem 1'' shows the maximum memory usage (in MB) during constructing $\mbf{F}_{\mrm{cho}}$ and $\mbf{I}_{\mrm{cho}}$ (obtained by calling \texttt{getmaxrss()} function after their construction finishes).
``Mem 2'' shows the maximum memory usage during constructing $\mbf{Z}_{\mrm{cho}}$ measured by a program whose input is $\mbf{I}_{\mrm{cho}}$ (that is, the usage does not include that of ``Mem 1''). ``OOM'' means out of memory (i.e., the memory usage exceeds 128GB).
We can confirm that our algorithm spent most of the time constructing $\mbf{Z}_{\mrm{cho}}$.

\begin{table}[t]
\begin{center}\small
  \caption{Running time (sec.) and memory usage (MB) for complete graphs with $n$ vertices. }
  \begin{tabular}{|r|r|r|r|r|r|} \hline
        & Const. $\mbf{F}_{\mrm{cho}}$ & Const. $\mbf{I}_{\mrm{cho}}$ & Const. $\mbf{Z}_{\mrm{cho}}$ & & \\
    $n$ & time & time & time & Mem 1 & Mem 2 \\ \hline\hline
2 & 0.000 & 0.000 & 0.000 & 28 & 28 \\ \hline
3 & 0.000 & 0.000 & 0.000 & 28 & 28 \\ \hline
4 & 0.001 & 0.000 & 0.001 & 28 & 28 \\ \hline
5 & 0.001 & 0.001 & 0.001 & 28 & 28 \\ \hline
6 & 0.002 & 0.003 & 0.004 & 28 & 28 \\ \hline
7 & 0.005 & 0.011 & 0.030 & 29 & 28 \\ \hline
8 & 0.014 & 0.041 & 0.333 & 29 & 32 \\ \hline
9 & 0.040 & 0.123 & 14.664 & 32 & 508 \\ \hline
10 & 0.109 & 0.496 & 692.666 & 45 & 15738 \\ \hline
11 & 0.323 & 1.574 & OOM & 92 & OOM \\ \hline

  \end{tabular}
  \label{tab:exp_time}
  \end{center}
\end{table}

Table~\ref{tab:exp_nodes} shows the number of non-terminal nodes of $\mbf{F}_{\mrm{cho}}$, $\mbf{I}_{\mrm{cho}}$ and $\mbf{Z}_{\mrm{cho}}$ and that of graphs (i.e., the cardinality of the family of sets represented by DDs) in $\mbf{F}_{\mrm{cho}}$ and $\mbf{Z}_{\mrm{cho}}$.
Note that the cardinality of the family
represented by $\mbf{F}_{\mrm{cho}}$ is the same as $\mbf{I}_{\mrm{cho}}$.
The cardinality of a family represented by a DD can be easily computed by a simple dynamic programming-based algorithm~\cite{Knuth11} in time proportional to the number of nodes in the DD.
The numbers appearing in the column ``\# Chordal labeled Graphs'' coincide those in the sequence A058862 in OEIS~\cite{Sloane_theencyclopedia}.

\begin{table}[t]
\begin{center}\small
  \caption{Number of nodes of 2-DDs and 3-DDs and that of graphs represented by the 2-DDs. ``\# cycles'' means that the number of cycles with length at least four.}
  \begin{tabular}{|r|r|r|r|r|r|} \hline
        & Const. $\mbf{F}_{\mrm{cho}}$ & Const. $\mbf{I}_{\mrm{cho}}$ & Const. $\mbf{Z}_{\mrm{cho}}$ & & \# Chordal \\
    $n$ & \# node & \# node & \# node & \# cycles & labeled graphs \\ \hline\hline
2 & 0 & 1 & 1 & 0 & 2 \\ \hline
3 & 0 & 4 & 3 & 0 & 8 \\ \hline
4 & 12 & 23 & 17 & 3 & 61 \\ \hline
5 & 54 & 176 & 106 & 27 & 822 \\ \hline
6 & 202 & 921 & 849 & 177 & 18154 \\ \hline
7 & 717 & 4883 & 8768 & 1137 & 617675 \\ \hline
8 & 2483 & 21959 & 111520 & 7962 & 30888596 \\ \hline
9 & 8569 & 119624 & 1736915 & 62730 & 2192816760 \\ \hline
10 & 29884 & 498703 & 32470737 & 555894 & 215488096587 \\ \hline
11 & 105789 & 2324022 & OOM & 5487894 & OOM \\ \hline

  \end{tabular}
  \label{tab:exp_nodes}
  \end{center}
\end{table}

\bibliographystyle{plain}

\end{document}